\documentclass{amsart}

\usepackage[T1]{fontenc}
\usepackage{times}
\usepackage[scaled=0.92]{helvet}
\usepackage{bbm,bm,amsmath,amsthm}
\usepackage[foot]{amsaddr}
\usepackage{dsfont}
\usepackage{graphicx,color}

\newtheorem{theorem}{Theorem}[section]
\newtheorem{definition}[theorem]{Definition}
\newtheorem{lemma}[theorem]{Lemma}
\newtheorem{proposition}[theorem]{Proposition}

\newcommand{\beq}{\begin{equation}}
\newcommand{\eeq}{\end{equation}}
\newcommand{\beqa}{\begin{eqnarray}}
\newcommand{\eeqa}{\end{eqnarray}}

\def\idty{\mathbbm{1}}
\def\tr{\mathrm{tr}}

\def \choi {\xi}

\def \C {{\mathbb C}}

\def \R {{\mathbb R}}
\def \S {{\mathcal S}}
\def \Z {{\mathbb Z}}

\def \H {{\mathcal H}}
\def \M {{\mathbf M}}

\def \N {\mathbb N}
\def \id {\mathrm {id}}

\newcommand\ketbra[2]{\ensuremath{%
    \vert{#1}\mkern1.2mu\rangle\langle\mkern1.2mu{#2}\vert}}
\newcommand\ket[1]{\ensuremath{%
    \vert{#1}\mkern1.2mu\rangle}}
\newcommand\bra[1]{\ensuremath{%
    \langle{#1}\mkern1.2mu\vert}}
\newcommand\hilbertH{\ensuremath{\mathcal{H}}}

\newcommand\hilberteins{\ensuremath{\mathbbm{1}}}

\begin{document}

\markboth{Andre Ahlbrecht, Florian Richter, and Reinhard F. Werner}
{Finite Roots of the Completely Depolarizing Channel}

\title{How long can it take for a quantum channel to forget everything?}

\author{Andre Ahlbrecht}
\address{Institute for Theoretical Physics, Leibniz Universit\"at Hannover,
Appelstra\ss e 2\\
30167 Hannover,
Germany}
\email{andre.ahlbrecht@itp.uni-hannover.de}

\author{Florian Richter}
\email{frichter@itp.uni-hannover.de}

\author{Reinhard F. Werner}
\email{reinhard.werner@itp.uni-hannover.de}

\maketitle

\begin{abstract}
We investigate quantum channels, which after a finite number $k$ of repeated applications erase all input information, i.e., channels whose $k$-th power (but no smaller power) is a completely depolarizing channel. We show that on a system with Hilbert space dimension $d$, the order is bounded by $k\leq d^2-1$, and give an explicit construction scheme for such channels. We also consider strictly forgetful memory channels, i.e., channels with an additional input and output in every step, which after exactly $k$ steps retain no information about the initial memory state. We establish an explicit representation for such channels showing that the same bound applies for the memory depth $k$ in terms of the memory dimension $d$.
\end{abstract}

\section{Introduction}	

Quantum channels are the mathematical description for the most general quantum information processing operations. In this paper we consider the question how quantum information can be erased in an iterated process. In the simplest case, all information is lost after a single step of this process, i.e. the quantum channel completely depolarizes its initial state.
Of course, the more interesting case is when the channel representing the one-step process acts non-trivially on the input system, but after a finite number of iterations leaves no information about the input system's initialization.
We refer to such a channel as a root of a completely depolarizing channel (see figure~\ref{fig:root}). One of the main objectives in this article is to show how long it can possibly take until such an iteration is completely depolarizing. An upper bound in terms of the system's dimension can be derived easily from the Jordan normal form of the channel, but since the process needs to represent a physical transformation, which is expressed by complete positivity of the corresponding map, it is not clear a priori whether this bound is attained by some channel. In order to show that this is indeed a tight bound, we develop an explicit construction scheme for maximal roots of completely depolarizing channels.

One motivation to look at this problem stems from quantum memory channels \cite{memorychannel}. These are channels which account for correlations between successive uses of the channel by introducing an additional system, referred to as the memory. A central question in this context is whether the influence of a fixed input on the memory dies out in time, i.e. whether the channel is {\em forgetful} or not \cite{memorychannel,rybar}. Moreover, if the impact of the memory input vanishes within a finite number of steps the channel is referred to as {\em strictly forgetful}. We will demonstrate a connection between the concept of roots of completely depolarizing channels and strictly forgetful memory channels. The idea is to consider the transformation of the memory as a function of the state of the external input system. When the memory channel is strictly forgetful, this must be a root of a completely depolarizing channel for all system states. The converse, however, is not true: There are memory channels which give a root of a completely depolarizing channel for all fixed inputs, but are not strictly forgetful for general sequences of possibly entangled input states. Nevertheless, by adapting our method to the setting of strictly forgetful memory channels it is possible to create a technique which yields all strictly forgetful memory channels.

\begin{figure}[htbp]
	\centering
		\includegraphics{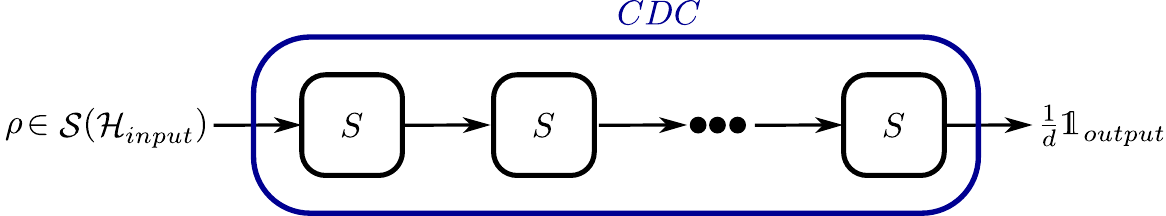}
	\caption{The channel $S$ is a finite root of the completely depolarizing channel (CDC) since a finite number of iterations maps an arbitrary input to the maximally mixed one ($d$ denotes the dimension of the quantum system).}
	\label{fig:root}
\end{figure}

As a second application of our theory, we discuss the generation of finitely correlated spin-chain states \cite{fcs} by a root of completely depolarizing channels. If the channel is an $k^{\rm th}$ root, one obtains so-called $k$-dependent states \cite{Petz,Matus}, which are defined by the property that the output observables separated by more than $k$ sites are independent.

Our paper is organized as follows. We first set up some notation and background on quantum channels on finite dimensional systems. In Section~\ref{sec:roots} we derive the general upper bound, and then describe the construction of maximal roots. For qubits we give an exhaustive construction, and focus for general systems on the question how the Jordan structure of a maximal root can be realized by completely positive maps. In Section~\ref{sec:fcs} we show how to obtain $k$-dependent states, and in Section~\ref{sec:memch} we discuss the connection to memory channels. 

\section{Quantum channels on finite dimensional systems}
The purpose of this section is to introduce notation and to give some necessary background on the mathematical aspects of quantum channels. For a detailed introduction to this topic we refer the reader to Paulsen's book \cite{paulsen}.

Throughout this paper we deal exclusively with quantum systems which can be described by a finite dimensional Hilbert space $\hilbertH=\C^d$. By $\M_d$ we denote the set of linear operators on $\C^d$ and the physical states of the system are represented by the set of density operators $\mathcal{S}(\C^d):=\{\rho\in\M_d,\ \rho\geq 0, \ \text{tr}(\rho)=1\}$. Possible measurements on the system are associated with the set of hermitian operators $\mathcal{M}(\C^d):=\{A\in\M_d,\ A^*=A \}$, where $A^*$ denotes the adjoint of $A\in\M_d$. A quantum channel can be defined in two different ways, we can either regard it as a transformation of the physical states or as a transformation of the measurements. The first point of view, also known as the Schrödinger picture, corresponds to a linear mapping $T^*$ from the states on an input system $\hilbertH_{in}$ to states of an output system $\hilbertH_{out}$, that is,
\begin{equation}
T^*:\mathcal{S}(\hilbertH_{in})\mapsto \mathcal{S}(\hilbertH_{out})\, .
\end{equation}
In the Heisenberg picture, which is precisely the second point of view, the quantum channel is represented by a linear map $T$ from the measurements on $\hilbertH_{out}$ to measurements on $\hilbertH_{in}$, i.e.,
\begin{equation}
T:\mathcal{M}(\hilbertH_{out})\mapsto \mathcal{M}(\hilbertH_{in})\, .
\end{equation}
The maps $T$ and $T^*$ are equivalent representations of the same physical transformation iff all expectation values for measurements $A$ performed on states $\rho$ coincide after application of $T$ respectively $T^*$. Hence, the representations in Heisenberg and Schrödinger picture are connected by the duality relation
\begin{equation}
\tr (T^*(\rho)A)=\tr(\rho\ T(A))\,.
\end{equation}
By linearity both maps extend to the whole space $\M_{d_{in}}$ respectively $\M_{d_{out}}$. In order to represent physical transformations, the maps $T^*$ and $T$ have to satisfy certain properties. Both have to be completely positive, that is, the extended maps $T^*\otimes \id_n$ and $T\otimes \id_n$, where $\id_n$ denotes the identity map on $n$-dimensional matrices, preserve positivity of operators. Additionally, it is often assumed that the quantum channel always generates an output when it is fed with a state of the input system. Mathematically, this is expressed by the assumption that $T^*$ is trace-preserving and $T$ is unital. Note that one property is a consequence of the other and the duality relation.

A channel which maps every input state to the same output state $\sigma$ is called a completely depolarizing channel (CDC). We denote the Schrödinger picture representation of this channel by $T_\sigma^*$, its mathematical definition reads $T^*_{\sigma}(\rho)= \tr (\rho) \sigma $, with $\sigma\in\mathcal{S}(\mathcal{H}_{out})$. The duality relation yields
\begin{equation}
\tr (T_{\sigma}^*(\rho)A)=\tr (\rho )\tr (\sigma A)=\tr ({\rho  \, \tr (\sigma A)\cdot \idty}),
\end{equation}
thus, $T_{\sigma}(A)=\tr {\sigma A}\cdot \idty$ for every $A\in\M_{d_{out}}$ is the representation of a CDC in the Heisenberg picture. The particular case where $d_{in}=d_{out}=d$ and $\sigma=\frac{1}{d}\idty$ leads to the CDC which is defined by $T^*_{\idty/d}(\rho)=\tr (\rho)\frac{1}{d}\idty$ and $T_{\idty/d}(A)=\tr (A)\frac{1}{d}\idty $. We will refer to this channel as the \emph{bistochastic} CDC. \\

The mathematical theory of completely positive maps provides some important results leading us to different ways of specifying a quantum channel. Since we are interested in concatenable channels, we focus our attention in the following to channels with equal input and output system. The first statement is the famous theorem of Kraus \cite{kraus}, which proves that every completely positive map $T$ admits a decomposition of the form
\begin{equation}
T(X)=\sum_{\alpha}K_{\alpha}^* X K_{\alpha}\quad \text{with}\quad \sum_{\alpha}K_{\alpha}K_{\alpha}^*=\mathds{1}\, .
\end{equation}
We refer to $\{K_{\alpha}\}$ as \emph{Kraus operators} of the channel $T$. Note that this representation involves a unitary degree of freedom, i.e., the channels defined by $\{\tilde{K_i}\}$ and $\{K_i\}$ coincide, if there exists a unitary $U$ such that $K_i=\sum_jU_{ij}\tilde{K_j}$ holds. If we restrict to Kraus decompositions with minimal numbers of Kraus operators, this is actually the only freedom we have in choosing the $K_i$. In other words, two minimal Kraus representations $\{\tilde{K_i}\}$ and $\{{K_i}\}$ are always connected by a unitary $U$ and the formula $K_i=\sum_jU_{ij}\tilde{K_j}$. Of course, two Kraus decompositions of a quantum channel $T$ do not necessarily consist of the same number of Kraus operators. For example, consider the convex combination $T=\lambda T_1+(1-\lambda)T_2$ of two channels $T_1$ and $T_2$ with Kraus operators $\{K_{1,i}\}$ respectively $\{K_{2,i}\}$. Clearly, $T$ can be written as a Kraus decomposition with operators $\{\sqrt{\lambda}K_{1,i}\}\bigcup \{\sqrt{1-\lambda}K_{2,i}\}$ but in general there exists a Kraus decomposition with fewer operators. We define the minimal number of Kraus operators as the \emph{Kraus rank} of the channel $T$ and note that a Kraus decomposition is minimal iff the operators $K_i$ are linearly independent. We will see in the next section that a possible Kraus decomposition of the bistochastic CDC in the case of a qubit input and output system is $T_{\idty/2}(X)=\frac{1}{4}\sum_{i=0}^{4}\sigma_i X\sigma_i$ with Pauli operators $\sigma_i$. In fact, this result can trivially be extended to higher dimensions by replacing Pauli operators by Weyl operators. This implies that the Kraus rank of the bistochastic CDC is always $d^2$, where $d$ is the system's dimension.\\

Another characterization arises if we consider Choi's theorem \cite{choi}. The statement of the theorem is sometimes called the channel-state duality, since it predicates a map $T$ is completely positive iff its corresponding \emph{Choi operator} $\choi_T:=T\otimes \id(\ket{\Omega}\bra{\Omega})=\frac{1}{d} \sum_{i,j}T(\ket{i}\bra{j})\otimes\ket{i}\bra{j} $, with $\ket{\Omega}=\frac{1}{\sqrt{d}}\sum_{i}\ket{ii}$, is positive. In fact, the trace of $\choi_{T^*}$ is normalized, hence, $T$ is completely positive iff $\choi_{T^*}$ is a state, we will refer to $\choi_{T^*}$ as the \emph{Choi state} of $T$. Since the relation between a channel $T$ and its corresponding Choi state $\choi_{T^*}$ is invertible, any state fully determines a channel and vice versa. The Choi state and Choi operator of a CDC are then given by
\begin{equation}
\choi_{T_\sigma^*}=\sigma\otimes \frac{1}{d}\mathbbm{1}_d\quad \text{and}\quad \choi_{T_\sigma}= \frac{1}{d}\mathbbm{1}_d\otimes \sigma^T\, ,
\end{equation}
where $\sigma^T$ denotes the transpose of $\sigma$. Furthermore, the linearity of a channel $T$ allows for representation of $T$ by a matrix $D_{T}$. For that reason we equip the vector space $\M_d$ with the Hilbert-Schmidt scalar product $\langle A\vert B \rangle_{HS}:= \tr (A^*B)$ and define the representation matrix of a channel as ${D_{T}}_{i,j}:=\langle A^i |T( A_j) \rangle_{HS}$ with $\{A_1,...,A_{d^2}\}$ as operator basis and $\{A^1,...,A^{d^2}\}$ its dual basis defined via $\tr ({A^i}^*A_j)=\delta_{i,j}$. We point out that usually there is no way to determine the complete positivity of a map solely from its representation matrix without any knowledge of the operator basis. However, if we choose the matrix units $E_{ij}:= \ket{i}\bra{j}$ with $i,j \in \{1,...,d\}$ as a basis for the representation we find that the representation matrix of $T$ and its corresponding Choi operator $\choi_T$ are connected via $D_{T_{nm,kl}}=\bra{E_{nm}} T(E_{kl})\rangle=d\bra{n\otimes k}\choi_T\ket{m\otimes l}$. Thus, the representation matrix of a certian CDC is given in this basis by
\begin{equation}\label{eq:CDC_matrix}
\bra{E_{nm}} T_\sigma(E_{kl})\rangle=\tr (E_{mn}\tr (\sigma E_{kl})\mathbbm{1})=\frac{1}{d}\delta_{n,m}\bra{l}\sigma \ket{k}.
\end{equation}

The divisibility of quantum channels, that is, the property of a channel $T$ to be decomposable into two non-trivial channels $S_1$ and $S_2$ has been investigated in reference \cite{dividingquantum}. In the paper at hand, we specify this investigations for the divisibility of a completely depolarizing channels $T_{\sigma}$ into a self-concatenation of identical maps, i.e. whether there exists a channel $S$ and $k\in \N$ such that $S^k=T_{\sigma}$. Finally, we draw some connections of this problem to other fields in quantum information theory.

We close this section with a mathematical definition of a $k^{\rm th}$ order root of a quantum channel. According to figure \ref{fig:root} we define:
\begin{definition}[Root of a Channel] \label{def:root}
A \textbf{$k^{\rm th}$ root} of the channel $T:\M_d\mapsto \M_d$ is a channel $S:\M_d\mapsto \M_d$ with
\begin{equation}
S^k=T \ \ \text{and} \ S^r\neq T \ \text{for} \  r<k, \quad k,r\in\mathbb{N}.
\end{equation}
We refer to $k$ as the \textbf{order} of a root.
\end{definition}

\section{Roots of Completely Depolarizing Channels}\label{sec:roots}

This section is started with some general comments about the construction of roots of a CDC. Our aim is to get roots with maximal number of necessary self-concatenations, i.e. maximal order roots. It will turn out that this maximal order is always $d^2-1$, where $d$ denotes the dimension of the underlying Hilbert space $\H$. If we consider the bistochastic CDC it is always possible to construct a maximal order root of this channel. After characterizing all maximal roots of the bistochastic CDC for qubit systems we present a construction scheme leading to maximal roots of the bistochastic CDC in arbitrary dimensions.

\subsection{General upper bound}
\label{Sec:GenApp}

The aim of this section is to describe the general approach we take to construct roots of a CDC. In particular, we investigate what the highest possible order of a finite CDC-root in terms of the dimension of $\H$ can be. Comparing the three introduced representations it turns out that the matrix-representation of a channel is the most fruitful one to determine an upper bound for the maximal order of a root.
\begin{theorem}[Boundedness of the root order]
\label{thm:bound}
Let $S:\M_d\mapsto \M_d$ be a channel, which is a $k^{\rm th}$ root of a completely depolarizing channel $T_\sigma$. Then $k\leq d^2-1$.
\end{theorem}
Before we prove the theorem, we need the following statement from linear algebra.
\begin{lemma}[Jordan normal form]\label{thm:jordan} \cite{LinearAlgebra}
For every matrix $T\in \M_D$ there is an invertible matrix $R$, such that
\begin{equation}
T=R\left(\bigoplus\limits_{\ell=1}^K J_\ell(\lambda_\ell)\right) R^{-1}
 =RJR^{-1},
\end{equation}
with
\begin{equation}
 J_\ell(\lambda):=\left(
                       \begin{array}{cccc}
                       \lambda & 1 &  &  \\
                        &\ddots  &\ddots  &  \\
                        &  & \lambda & 1 \\
                        &  &  & \lambda \\
                       \end{array}
                     \right)\in \M_{d_\ell},
\end{equation}
where the matrix $J$ is called the \textbf{Jordan normal form} of $T$ and the $J_\ell(\lambda)$ are the \textbf{Jordan blocks} of size $d_\ell$ corresponding to the eigenvalue $\lambda$. The number of Jordan blocks with the eigenvalue $\lambda$ is the \textbf{geometric multiplicity} of it while the sum of the dimensions $\sum_{\lambda=\lambda_\ell}d_\ell$ is the \textbf{algebraic multiplicity} of the eigenvalue.
\end{lemma}
If we express the property of $S$ being a root of a CDC in terms of the representation matrix, we can establish the upper bound in the following way:

\begin{proof}[Proof of Theorem~\ref{thm:bound}]%
We consider $T_\sigma$ and $S$ as operators on $\M_d$, a space of dimension $D=d^2$.
The eigenvalues of the CDC channel $T$ are $1$ and $0$. Hence for any eigenvalue $\lambda$ of $S$ we have $\lambda^k\in\{0,1\}$, so $S$ likewise has only the eigenvalues $1$ and $0$. Since $1$ is a simple eigenvalue of $T$, the eigenvalue $1$ of $S$ is also simple, and there will be no roots of unity. Hence the algebraic multiplicity of the eigenvalue $0$ of $S$ is $d^2-1$. The only remaining question is the decomposition of this dimension into Jordan blocks. The root order will be the smallest $k$ such that $J_\ell(0)^k=0$ for all $\ell$. The smallest power $k$ for which $J_\ell(0)^k=0$ is $d_\ell$. Hence the root order is the dimension of the largest Jordan block. Clearly, this becomes largest when there is only one block, i.e., when $k=d^2-1$.
\end{proof}

Although the above theorem gives an explicit upper bound for the order of a root of a CDC it does not answer the question whether this bound is attained by some channel $S$. In order to reach this bound we additionally need to care about complete positivity of $S$, which cannot be decided solely from the representation matrix of $S$. However, we will see in the next sections that there always exist channels which attain the upper bound $d^2-1$ for the order of a root of the bistochastic CDC. Since for the bistochastic CDC $T_{\idty/d}$ and $T_{\idty/d}^*$ are represented by the same map, we omit the $^*$ in the notation to distinguish between Heisenberg and Schrödinger picture for the rest of this section.

\subsection{All roots of the bistochastic qubit CDC }
\label{Sec:Qubit}

The most elementary case arises if we consider the input and output system to be qubits. According to theorem \ref{thm:bound}, the highest possible order of a CDC-Root is three in this case. We will give an explicit characterization of the whole set of maximal qubit roots if $T_\sigma$ is the bistochastic CDC, i.e. $\sigma=\frac{1}{d}\idty$. Thereby we verify that there are indeed roots of the CDC of order three.\\
In reference \cite{mary1} it is proven that every bistochastic qubit channel can be decomposed as
\begin{equation}\label{eq:qubit_case}
	T(\rho)=U_1\Lambda[U_2\rho U_2^*]U_1^*
\end{equation}
where $U_1$ and $U_2$ are unitaries and $\Lambda$ is a Pauli diagonal channel, i.e. $\Lambda[\sigma_i]=\lambda_i\sigma_i\ \text{with}\ i\in\{0,1,2,3\}$. If we represent the set of possible qubit states $\rho$ via the Bloch sphere $\rho(\vec{r})=\frac{1}{2}(\hilberteins+\vec{r}\vec{\sigma}), \Vert\vec{r}\Vert\leq 1$ and translate the action of the maps induced by $U_1,U_2$ and $\Lambda$ to maps acting on $\vec{r}$, we find that an arbitrary qubit channel can be written as a composition of a diagonal map $\mathbf{L}=\text{diag}\{\lambda_1,\lambda_2,\lambda_3\}$ and two rotations:
\begin{equation}\label{eq:bloch}
T(\rho(\vec{r}))=\frac{1}{2}(\hilberteins+(\mathbf{R_1 L R_2}\vec{r})\vec{\sigma}), \ \mathbf{R_i}\in SO(3)
\end{equation}
Clearly, since the $\mathbf{R_i}$ are invertible, the rank of the linear map $T$ is determined by the choice of the values $\{\lambda_i\}$ of the diagonal map $\mathbf{L}$. Thus, the rank of $\mathbf{L}$ completely determines the order of a potential root $T$. Indeed, let us assume $T$ is a $k^{\rm th}$ order root of the CDC, then $k=3$ if the rank of $\mathbf{L}$ is two and $k=2$ if the rank is one, as can be seen from  the corresponding Jordan normal forms.

To explore the possible configurations of the $\{\lambda_i\}$ resulting in completely positive maps we need to have a closer look at the set the Pauli diagonal channels. It is a well-known fact that the set of possible Pauli diagonal channels form a tetrahedron \cite{mary2}. To get a suitable characterization of the rank of the channels, we use the fact that every Pauli diagonal channel has a Kraus decomposition of the form $T(X)=\sum_{i=0}^3\mu_i\sigma_i X \sigma_i$. Some straightforward calculations show that these channels are indeed Pauli diagonal with the relations
\begin{align}\label{eq:muuslambdas}
\lambda_0&=\mu_0+\mu_1+\mu_2+\mu_3 \\ \lambda_1&=\mu_0+\mu_1-\mu_2-\mu_3 \notag\\
\lambda_2&=\mu_0-\mu_1+\mu_2-\mu_3 \notag  \\ \lambda_3&=\mu_0-\mu_1-\mu_2+\mu_3 \notag
\end{align}
between $\{\mu_i\}$ and $\{\lambda_i\}$. If, on the other hand, we consider a map $T$ in terms of the $\lambda_i$, we can solve \eqref{eq:muuslambdas} for the $\mu_i$ to get a Kraus decomposition of $T$. It is easy to see that the eigenvectors of the Choi operator $\choi_T:=T\otimes \id(\ketbra{\Omega}{\Omega})$ are the vectors $\ket{\Omega_i}:=\hilberteins\otimes\sigma_i\ket{\Omega}$ with corresponding eigenvalue $\mu_i$. Hence, complete positivity of $T$ is expressed by the condition that all $\mu_i$ are positive numbers and unitality of $T$ requires that the $\mu_i$ add up to one, i.e. $\lambda_0=1$. This yields the inequalities
\begin{align}\label{eq:tetrahedron}
\lambda_1+\lambda_2+\lambda_3& \geq -1\\
\lambda_1-\lambda_2-\lambda_3&\geq -1\notag\\
-\lambda_1+\lambda_2-\lambda_3& \geq -1\notag  \\
-\lambda_1-\lambda_2+\lambda_3& \geq -1\notag
\end{align}
characterizing the tetrahedron formed by the set of admissible $\{\lambda_i\}$. These relations imply that the set of Pauli diagonal channels with one eigenvalue equal to zero can be represented by squares inside the tetrahedron (see figure \ref{fig:RootQubitChannel}).\\

\begin{figure}
	\centering
		\includegraphics{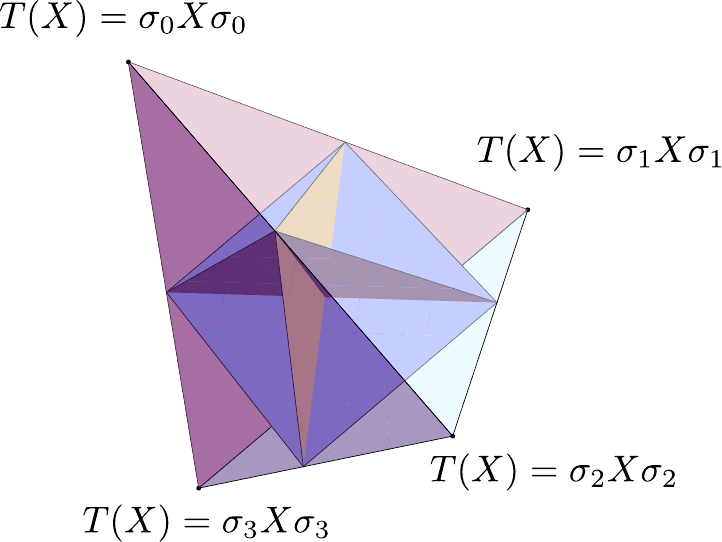}
	\caption{The tetrahedron of Pauli diagonal channels parameterized by the eigenvalues $\lambda_i$ with $i=1,2,3$. The extremal points of the tetrahedron represent the configuration where exactly one of
the Kraus weights $\mu_i$ equals one. The squares inside the tetrahedron mark the configurations, where one of the eigenvalues of the channel is equal to zero. The bistochastic CDC is represented by the intersection point of all three squares.}
	\label{fig:RootQubitChannel}
\end{figure}

Furthermore, we find an interesting connection between the Kraus rank of a channel and the rank as a linear map for qubit channels. Starting at an extremal point of the tetrahedron we find a reversible channel with Kraus rank one. If we move on the line between two extremal points we increase the Kraus rank by one under preservation of the reversibility. In the middle of the line the rank is lowered by two, because it lies on the intersection of two squares. Hence we need at least three Kraus operators to construct a qubit channel with rank exactly one less than maximal. Since the Pauli diagonal map characterizes the rank completely this holds for all bistochastic qubit channels.\\

The following theorem gives a complete characterization of maximal roots for qubit systems:
\begin{theorem}[Maximal qubit roots] A bistochastic channel $T:\mathcal{S}(\C^2)\mapsto \mathcal{S}(\C^2)$ is a maximal root of the bistochastic CDC, iff $T$ is of the form
\begin{equation}
T(\rho(\vec{r}))=\frac{1}{2}(\hilberteins+(\mathbf{R_2^{T}R_1 L R_2}\vec{r})\vec{\sigma}), \ \mathbf{R_i}\in SO(3)
\end{equation}
where $\mathbf{L}$ has exactly two non-zero eigenvalues,  $\mathbf{R_2}\in SO(3)$ is an arbitrary rotation and there exist angles $\phi, \theta \in [0,2\pi)$ such that
\begin{equation}
\mathbf{R_1}=
\left(
\begin{array}{ccc}
0 & \cos \theta & -\sin \theta \\
\sin \phi & \cos \phi \sin \theta & \cos \phi \cos \theta \\
-\cos \phi & \sin \phi \sin \theta & \sin \phi \cos \theta
\end{array}
\right)\,  .
\end{equation}
If we choose $\mathbf{L}=\text{diag}({0,\lambda_2, \lambda_3})$ we have the restrictions $\vert \lambda_2\pm\lambda_3\vert\leq 1$ for complete positivity and $\tan \phi =-\frac{\lambda_2}{\lambda_3}\tan \theta $.
\end{theorem}
\begin{proof}
As already mentioned, for a maximal root of the bistochastic CDC we have to choose the Pauli diagonal map $\mathbf{L}$ in \eqref{eq:bloch} such that it has exactly two non-zero eigenvalues. Without loss of generality we assume $\mathbf{L}=\text{diag}({0,\lambda_2, \lambda_3})$. In the following, we determine all roots of the bistochastic CDC whose decomposition \eqref{eq:bloch} is such that $\mathbf{R_2}=0$. The general case reduces to this particular setting by the following argument: If $\mathbf{R_1L R_2}$ is a general root of the bistochastic CDC we have
\begin{equation}
\mathbf{0}=(\mathbf{R_1 L R_2})^n= \mathbf{R_2}^{-1}(\mathbf{ R_2 R_1 L})^n \mathbf{R_2}\, ,
\end{equation}
which means that $\mathbf{ R L}$, with $\mathbf{ R}=\mathbf{ R_2 R_1}$, is a root as well. On the other hand, if $\mathbf{R_1L}$ is a root then
\begin{equation}
\mathbf{0}=(\mathbf{R_1 L})^n=\mathbf{R_2}(\mathbf{R_2}^{-1}\mathbf{R_1 L R_2})^n \mathbf{R_2}^{-1}\, ,
\end{equation}
and hence $\mathbf{R_2}^{-1}\mathbf{R_1 L R_2}$ is again a root.

Hence, the task is to determine all three-dimensional orthogonal matrices $\mathbf{R}$ such that the composition

\begin{equation}
	D_T=\mathbf{R}\cdot \text{diag}({0,\lambda_2, \lambda_3})
\end{equation}
is nilpotent. This is equivalent to saying that all eigenvalues of $D_T$ equal zero. Since $D_T$ is of the form
\begin{equation}
D_T=
\left(
\begin{array}{ccc}
0 & \lambda_2 r_{12} & \lambda_3 r_{13} \\
0 & \lambda_2 r_{22} & \lambda_3 r_{23} \\
0 & \lambda_2 r_{32} & \lambda_3 r_{33}
\end{array}
\right)
\end{equation}
its characteristic polynomial can be written as $\chi (z)= \det (D_T-z\mathbbm{1})=-z\cdot (z^2 - z\text{tr} ( \Lambda) +\det ( \Lambda))$, where we introduced the submatrix
\begin{equation}
\Lambda=
\left(
\begin{array}{cc}
 \lambda_2 r_{22} &  \lambda_3 r_{23} \\
 \lambda_2 r_{32} &  \lambda_3 r_{33}
\end{array}
\right)\,.
\end{equation}
Hence, the condition that all three eigenvalues of $D_T$ equal zero is equivalent to the condition:
\begin{equation} \label{eq:condition}
\det (\Lambda))=0 \quad \text{and} \quad \text{tr} (\Lambda))=0
\end{equation}
The first part of the condition is already satisfied if the 2x2 submatrix of $\mathbf{R}$ consisting of $\{r_{22},r_{23},r_{32},r_{33}\}$ has a vanishing determinant. Therefore we choose the following ansatz:
\begin{equation}
\mathbf{R} =
\left(
\begin{array}{ccc}
 & \hdots & \\
\vdots &  a &  b \\
 &  z\cdot a &   z\cdot b
\end{array}
\right)\, ,\quad a,b,z \in \mathbbm{R}
\end{equation}
We get additional restrictions on $\{a,b,z\}$ from the orthogonality of $\mathbf{R}$, that is, the rows of $\mathbf{R}$ have to be normalized and mutually orthogonal. This leads to the conditions
\begin{eqnarray}
a^2+b^2&<&1 \\
z^2(a^2+b^2)&<&1 \nonumber
\end{eqnarray}
and
\begin{equation}
(z^2+1)(a^2+b^2)=1 \, .
\end{equation}
These relations suggest a parametrization in terms of trigonometric functions. If we expand the rows and columns accordingly, we find that
\begin{equation}
\mathbf{R}=
\left(
\begin{array}{ccc}
0 & \cos \theta & -\sin \theta \\
\sin \phi & \cos \phi \sin \theta & \cos \phi \cos \theta \\
-\cos \phi & \sin \phi \sin \theta & \sin \phi \cos \theta
\end{array}
\right)\, \text{with} \ \phi,\theta \in \mathbbm{R}
\label{eq:rot}
\end{equation}
is a possible parametrization of all $\mathbf{R}\in SO(3)$ which fulfill the first condition of \eqref{eq:condition}. The second condition fixes a relation between the two variables, such that $\theta$ can be an arbitrary angle and $\phi$ has to satisfy $\tan \phi =-\frac{\lambda_2}{\lambda_3}\tan \theta $.

\end{proof}

\subsection{Roots via perturbation}
\label{roots}

Inspired by the approach of the last section, one might try to imitate the construction of a maximal root via the composition of a rank-lowering and rotation maps in higher dimensions. For instance, one could replace the Pauli matrices by Weyl operators and consequently compose a rank-lowering Weyl diagonal map with some unitary channel representing the rotation. This procedure fails mainly due to a lack of a Bloch-sphere interpretation of the state space and the dimensional gap between $SU(d)$ and $SO(d^2-1)$, i.e., the unitary channels do not cover all possible $SO(d^2-1)$-rotations of the state space in dimensions beyond $d=2$.

Nevertheless, our aim is to construct maximal roots of the bistochastic CDC for arbitrary system dimension $d$. For this purpose, we recall the fact, that the Jordan normal form of a maximal root is, in an appropriate basis, given by the direct sum of a projector on the maximally mixed state $\sigma=\frac{1}{d}\idty$, representing the bistochastic CDC, and a maximal Jordan block to the eigenvalue zero. The key idea is now to consider this nilpotent Jordan block as a $\varepsilon$-weighted perturbation of the bistochastic CDC in such a way that the complete positivity remains untouched. The following theorem shows, that this is indeed possible.
\begin{theorem}
\label{thm:rootsperturb}
Let $B:=\{A_1=\mathbbm{1}_d,A_2,...,A_{d^2}\}$ be a basis of hermitian operators in $\M_d$ and $B^*=\{A^1=\frac{\mathbbm{1}_d}{d},A^2,...,A^{d^2}\}$ its dual basis, with respect to the Hilbert-Schmidt scalar product $\tr ({A^j}^*A_i)=\delta_{i,j}$. Then, for small enough $\varepsilon\in\R$, the map
\begin{equation}
T_{\varepsilon}(X)=\frac{1}{d}\idty \tr X+\varepsilon \sum_{i=2}^{d^2-1}A_i\tr ({A^{i+1}}^* X)\,,
\end{equation}
is completely positive and therefore a root of the bistochastic CDC of maximal order $d^2-1$.
\end{theorem}
\begin{proof}
First we emphasize that the dual basis $B^{*}$ is hermitian as well. To verify this, consider the $A_i$ as basis for the real vector space of all hermitian operators and construct the unique dual basis $A^j$ as a linear combination of the $A_i$ with real coefficients. The key point is to choose $\varepsilon$ in such a way that $T_{\varepsilon}$ becomes completely positive. A way to establish the existence of such an $\varepsilon$ is to consider the Choi operator $\choi_{T_\varepsilon}$ and choose $\varepsilon$ such that $\choi_{T_\varepsilon}$ is positive. A straightforward calculation yields
\begin{equation}
\choi_{T_{\varepsilon}}= \frac{1}{d^2}\idty + \frac{\varepsilon}{d}\sum_{i=2}^{d^2-1}A_i\otimes \overline{A^{i+1}}\,,
\end{equation}
where $\overline{A^i}$ denotes the complex conjugate of $A^i$ in a fixed basis of $\C^d$. Hence, the map $T_\varepsilon$ is completely positive if $\varepsilon$ is small enough to ensure
\begin{equation}\label{eq:epsbound}
- \idty \leq \varepsilon d \sum_{i=2}^{d^2-1}A_i\otimes \overline{A^{i+1}}=:\rho_\varepsilon\,.
\end{equation}
Since the chosen basis of $\M_d$ and its dual are hermitian this amounts to a comparison between the eigenvalues of the hermitian operators $\rho_\varepsilon$ and $\idty$. The continuity of the eigenvalue distribution of $\rho_\varepsilon$ assures the existence of a proper $\varepsilon$ for arbitrary $d$ and therefore guarantees the complete positivity of $T_{\varepsilon}$ for small enough $\varepsilon$.
\end{proof}
The construction scheme presented in the proof of theorem \ref{thm:rootsperturb} also applies for all CDCs $T_\sigma$ such that $\sigma$ has full rank. If, however, the rank $r$ of $\sigma$ is less than maximal this construction only leads to a root of order $r^2-1$.

Of course, for a given basis $A_i$ and its dual $A^i$ there is an explicit bound on $\varepsilon$ in terms of the eigenvalues of $\rho_\varepsilon$. This bound seems to be decreasing with the dimension $d$, due to the enlargement of the spectral radius in \eqref{eq:epsbound}, while going to larger dimensions. This suggests that all maximal roots of the bistochastic CDC in higher dimensions obtained by this construction get closer to the bistochastic CDC with increasing $d$. The following proposition refutes this statement:
\begin{proposition}
For arbitrary dimension $d\in\N$ there is always a maximal root $T_{\varepsilon}$ of the bistochastic CDC such that their cb-norm distance satisfies
\begin{equation}
\Vert T_{\varepsilon}-T_{\frac{1}{d}\idty}\Vert_{cb}\geq \frac{d-1}{d}\, .
\end{equation}
\end{proposition}

\begin{proof}
We start with a further specification of the basis $A_i$ and choose $\Vert A_3\Vert_\infty=1$. This implies the following lower bound for the cb-norm difference between $T_\varepsilon$ and the bistochastic CDC:
\begin{align}
\Vert T_\varepsilon-T_{\frac{1}{d}\idty}\Vert_{cb}&\geq \sup\limits_{\Vert A\Vert_{\infty}\leq1}\Vert T_\varepsilon(A)-T_{\frac{1}{d}\idty}(A)\Vert_{\infty}  \\
&\geq\Vert T_\varepsilon(A_3)-T_{\frac{1}{d}\idty}(A_3)\Vert_{\infty} \nonumber\\
&= \vert\varepsilon\vert\Vert A_2\Vert_{\infty} \nonumber
\end{align}
Hence, the cb-norm is bounded from below by $(d-1)/d$ if we are free to choose $\varepsilon \Vert A_2\Vert_{\infty}=(d-1)/d$. Of course, this choice should not violate the complete positivity of $T_\varepsilon$.

In order to show that this is indeed possible with some further restrictions to the choice of $A_3$, we consider the transformation $A_i\mapsto \delta^{-i+3} A_i$ and $A^i\mapsto \delta^{i-3}A^i$ with $\delta\in \R\backslash\{0\}$ and $i>3$, which still gives a maximal root of the bistochastic CDC. The Choi operator $\choi_{T_\varepsilon}$ transforms under this map according to
\begin{equation}
\choi_{T_\varepsilon} \mapsto \frac{1}{d^2}\idty + \frac{\varepsilon}{d}A_2\otimes \overline{A^{3}}+\delta\frac{\varepsilon}{d}\sum_{i=3}^{d^2-1}A_i\otimes \overline{A^{i+1}} \,.
\end{equation}
Since we get a maximal root of the bistochastic CDC for arbitrarily small but non-zero $\delta$, perturbation theory tells us that we only have to compare the eigenvalues of the first two parts of the sum when considering the limit $\delta\rightarrow 0$. In other words, we may neglect the influence of terms $A_i\otimes \overline{A^{i+1}}$ with $i>2$ on the eigenvalue problem if $\delta$ is chosen sufficiently small. Hence, we have to establish the inequality
\begin{equation}\label{eq:firsttwo}
-\idty \leq d \varepsilon A_2\otimes \overline{A^{3}}\,.
\end{equation}
By taking the operator norm on both sides of the inequality, we find the following sufficient criterion for complete positivity of $T_{\varepsilon}$:
\begin{equation}\label{eq:ineq}
|d||\varepsilon|\Vert A_{2}\Vert_{\infty}\Vert \overline{A^{3}}\Vert_{\infty}\leq 1
\end{equation}
 To obtain the restrictions imposed on $A_3$ by this inequality, we furthermore assume the basis $B$ to be orthogonal, that is, $A^i=A_i/\tr  A_i^2$. We denote the eigenvalues of $A_3$ by $a_3^i$, then, orthogonality to $A_1=\idty$ requires $\sum_{i}a_3^{i}=0$. Additionally, $A_3$ must satisfy $\max\limits_i\vert a^i_3\vert=1$  to guarantee $\Vert A_3\Vert_\infty=1$. With these notations \eqref{eq:ineq} changes into
\begin{equation}
|\varepsilon|\Vert A_{2}\Vert_{\infty}\leq\frac{|\tr  A_{3}^2|}{d}=\frac{\sum_i\vert a^i_3\vert^2}{d}.
\end{equation}
Thus, equation \eqref{eq:firsttwo} is satisfied and $T_{\varepsilon}$ is completely positive, if we choose $A_3$ such that $d-1\leq \sum_i\vert a^i_3\vert^2$. This can be satisfied in even dimensions by choosing $a_3^{i}=(-1)^{i}$, actually leading to a lower bound of $1$ for the cb-norm difference. In odd dimensions we choose $a_3^{i}=(-1)^{i}$ for $i<d$ and $a_3^d=0$.
\end{proof}

\section{Finitely correlated construction of $k$-dependent states}\label{sec:fcs}
The general concept of \emph{finitely correlated states} and the occurring correlations are considered in reference \cite{fcs}. We want to deal with the correlations that occur if a maximal CDC-root $S$ is used to generate a functional on the infinite spin chain. For this purpose we consider the quasi-local algebra $\mathcal{A}:=\bigotimes_{i=-\infty}^{\infty}\M_{d_i}$ generated by algebras of finite subsets $\mathcal{A}_\Lambda:=\bigotimes_{z\in\Lambda} \mathcal{A}_z$ on finite chain elements $\Lambda\subset \Z$. A $k$-dependent state $\omega$ is defined in the following way \cite{Petz,Matus}:
\begin{definition}[$k$-dependent state]
A state $\omega: \mathcal{A}\mapsto \C$ is called \textbf{$k$-dependent} if algebras separated by $k$ or more sites are independent, i.e.
\begin{equation} \label{eq:dependent}
\omega\left(A_{(-\infty,n)}\otimes\underbrace{\mathbbm{1}\otimes...\otimes\mathbbm{1}}_{k\text{-times}}\otimes A_{[n+k+1,\infty)}\right)=\omega\left(A_{(-\infty,n)}\right)\omega\left(A_{[n+k+1,\infty)}\right).
\end{equation}
\end{definition}
To understand the application of maximal CDC-roots in this context, we first recall that every channel $T:\M_{d'}\mapsto \M_{d}$ admits a Stinespring representation \cite{stinespring}, i.e.
\begin{equation}\label{eq:stinespring_dilation2}
T(X)=V^*(X\otimes\mathbbm{1}_k)V\ \ \forall X\in\M_{d'},
\end{equation}
where $V:\mathbb{C}^d\mapsto\mathbb{C}^{d'}\otimes\mathbb{C}^{k}$ is an isometry, that is $V^*V=\mathbbm{1}_d$. Since we want to concatenate the channel $T$, input dimension $d=\text{dim}(\hilbertH_{in})$ and output dimension $d'=\text{dim}(\hilbertH_{out})$ will be equal. Furthermore, the dimension $k$ of the ancilla system is equal to the Kraus rank of $T$. Keeping this representation in mind, we define the following map $\mathbbm{E}_{A}:\M_d\otimes \M_k\mapsto \M_d$:
\begin{equation}
\quad \mathbbm{E}_{A}(X)=V^*(X\otimes A)V\, ,
\end{equation}
where the isometry $V$ is chosen to be the same as in \eqref{eq:stinespring_dilation2} and therefore $\mathbbm{E}_{\mathbbm{1}}(X)=T(X)$ holds. The concatenation of several $\mathbbm{E}_{A_i}$, with $i=1,\ldots ,n$, together with any $\rho\in\mathcal{S}(\C^d)$ defines a functional $\omega_n:\M_k^{\otimes n}\mapsto \C$ via
\begin{equation}
\omega_n(A_1\otimes A_2 \otimes ... \otimes A_n)=\tr\left(\rho\mathbbm{E}_{A_1}\circ\mathbbm{E}_{A_2}\circ...\circ\mathbbm{E}_{A_n}(\mathbbm{1})\right).
\end{equation}
Due to the unitality of $T$, we can extend this functional to the positive half chain $\mathcal{A}_+:=\bigotimes_{i=0}^{\infty}\M_k$ via:
\begin{align}
\omega_{n+1}\left(A_{[1,n]}\otimes\mathbbm{1}\right)&=\tr\left(\rho\mathbbm{E}_{A_1}\circ\mathbbm{E}_{A_2}\circ...\circ\mathbbm{E}_{A_n}\circ\mathbbm{E}_{\mathbbm{1}}(\mathbbm{1})\right)\\
&=\tr\left(\rho\mathbbm{E}_{A_1}\circ\mathbbm{E}_{A_2}\circ...\circ\mathbbm{E}_{A_n}(\mathbbm{1})\right) \nonumber \\
&=\omega_n\left(A_{[1,n]}\right)\, .\nonumber
\end{align}
Furthermore, if we choose $\rho$ as an invariant state of $T$, i.e. $\tr\left(\rho\mathbbm{E}_{\mathbbm{1}}(X)\right)=\tr\left(\rho X\right)$, we can extend the functional also to the negative half-chain $\mathcal{A}_-:=\bigotimes_{i=-\infty}^{0}\M_k$ through setting:
\begin{align}
\omega_{n+1}\left(\mathbbm{1}\otimes A_{[1,n]}\right)&=\tr\left(\underbrace{\rho\mathbbm{E}_{\mathbbm{1}}}_{\rho}\circ\mathbbm{E}_{A_1}\circ\mathbbm{E}_{A_2}\circ...\circ\mathbbm{E}_{A_n}(\mathbbm{1})\right)\\
&=\omega_n\left(A_{[1,n]}\right)\, .\nonumber
\end{align}
Combining these two extensions we define a functional on the infinite spin chain $\mathcal{A}$. If we furthermore define the \emph{shift operator} $\sigma$ by setting
\begin{equation}
\sigma:\mathcal{A}\mapsto \mathcal{A} \quad,\quad \sigma(A_1\otimes...\otimes A_n\otimes\mathbbm{1})=\mathbbm{1}\otimes A_1\otimes\cdots\otimes A_n,
\end{equation}
we find that $\omega$ is a \emph{translation invariant state}, i.e. $\omega=\omega\circ\sigma$. We refer to $\omega$ as \emph{finitely correlated state} generated by $(T,\rho)$.

Now we choose the generating channel $T$ as a $k^{\rm th}$ root of the bistochastic CDC, that is, $\mathbbm{E}_{\mathbbm{1}}^{d^2-1}(X)=\frac{\mathbbm{1}}{d}\tr X$. Our construction scheme for roots of the bistochastic CDC yields channels $T$ for which the maximally mixed state is an invariant state, hence, we choose $\rho=\frac{\mathbbm{1}}{d}$. This means we construct a finitely correlated state on an infinite spin chain with a certain dependency length. This length is equal to the order of the root, as the following calculation for the resulting functional $\omega$ shows:
\begin{alignat}{1}
&\omega\left(\ldots \otimes A_n\otimes\underbrace{\mathbbm{1}\otimes \ldots\otimes\mathbbm{1}}_{d^2-1\text{-times}}\otimes A_{n+d^2}\otimes \ldots\right)\\ \notag
&=\tr\left(\rho  \ldots\circ\mathbbm{E}_{A_{n-1}}\circ\mathbbm{E}_{\mathbbm{1}}^{d^2-1}\circ\mathbbm{E}_{A_{n+d^2}}\circ\ldots (\mathbbm{1}) \right)\\ \notag
&=\tr\left(\rho\ldots \circ\mathbbm{E}_{A_{n-1}}(\mathbbm{1})\right)\tr\left(\rho\mathbbm{E}_{A_{n+d^2}}\circ\ldots (\mathbbm{1})\right)\\ \notag
&=\omega\left(A_{(-\infty,n)}\right)\omega\left(A_{[n+d^2,\infty)}\right).
\end{alignat}
This satisfies the form of \eqref{eq:dependent} and therefore the maximal roots generate $d^2-1$-dependent states on an infinite spin chain, with $d$ the dimension of $\rho$.

\section{Memory channels}\label{sec:memch}
Commonly it is assumed that successive uses of a channel are uncorrelated in the sense that identical inputs at different time steps produce identical outputs. However, almost all real physical processes exhibit some correlation in time, i.e. the transformation of the states at some time $t$ depends to some extent on the states at previous times $t'<t$. If we consider the repeated application of a quantum channel, e.g. sending photons through some fiber, these correlations can be taken into account by introducing an additional system $\mathcal{M}$, which we refer to as the \emph{memory system}.
\begin{figure}
\centering
	\includegraphics{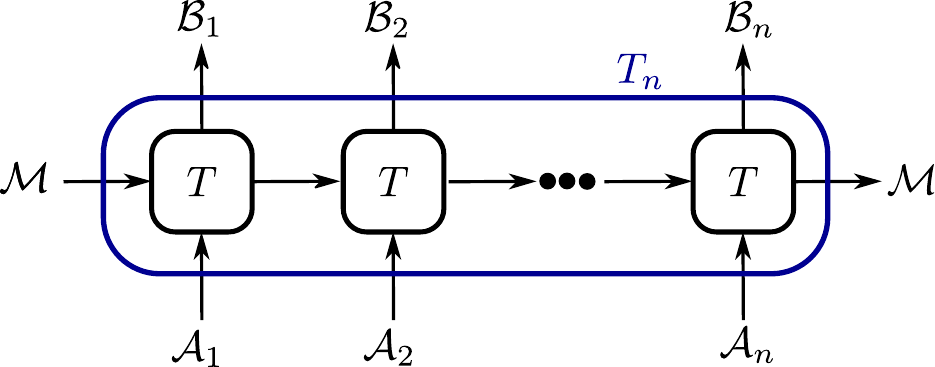}
	\caption{A quantum memory channel $T$ and its n-fold concatenation $T_n$. $\mathcal{M}$ denotes the memory systems which is used to model interaction between the different concatenation steps.}
	\label{fig:memory}
\end{figure}
For the sake of clarity we define $\mathcal{S}(\hilbertH_{in})=:\mathcal{A}$ for the input system and $\mathcal{S}(\hilbertH_{out})=:\mathcal{B}$ for the output system. Then, the $n$-fold concatenation $T_n: \mathcal{M}\otimes \mathcal{A}^{\otimes n}\mapsto  \mathcal{B}^{\otimes n}\otimes\mathcal{M}$ of the quantum memory channel $T:\mathcal{M}\otimes\mathcal{A} \mapsto  \mathcal{B}\otimes\mathcal{M}$ can be expressed via
\begin{equation}\label{eq:memory_concatenation}
T_n=\left(id_{\mathcal{B}}^{\otimes n-1}\otimes T\right)\circ..\circ\left(id_{\mathcal{B}}\otimes T\otimes id_{\mathcal{A}}^{\otimes n-2}\right)\circ\left(T\otimes id_{\mathcal{A}}^{\otimes n-1}\right)\, ,
\end{equation}
where $id_{X}: X\rightarrow X$ denotes the ideal or noiseless channel on system $X$, see figure \ref{fig:memory} for an illustration. Due to this construction, the elements of the output algebras $\mathcal{B}$ will certainly be affected by the choice of the initial memory state $\rho$ of the memory system $\mathcal{M}$. Indeed, it is a natural question to ask if all elements of the output system are influenced in the same way or if the effect of the memory dies out after sufficiently many concatenation steps. This leads to the notion of \emph{forgetful} memory channels, which are those quantum memory channels where the influence of the initialization of the memory systems vanishes exponentially with the number of time steps, see reference \cite{memorychannel} for a precise definition. Our aim in this section is to construct forgetful memory channels where the effect of the initial memory system vanishes completely after a certain number of concatenations. We refer to such channels as \emph{strictly forgetful} memory channels. The following definition expresses this in mathematical terms:
\begin{definition}[Strictly Forgetful Memory Channel]
A quantum memory channel $T$ is strictly forgetful, iff there is some $n\in\mathbbm{N}$, such that
\begin{equation}
\label{Eq:DefStricForget}
||\tr_{\mathcal{B}^{\otimes n}}\left[T_n((\sigma_{\mathcal{M},1}-\sigma_{\mathcal{M},2})\otimes\sigma_{sys})\right]||_{1}=0
\end{equation}
for all $\sigma_{\mathcal{M},1}$,\ $\sigma_{\mathcal{M},2} \in \mathcal{S}(\mathcal{M})$ and $\sigma_{sys}\in\mathcal{S}(\mathcal{A}^{\otimes n})$.
\end{definition}
This definition assumes that there is no entanglement between the initial memory state and the system. We refer to the minimal $n$, such that \eqref{Eq:DefStricForget} holds, as \emph{memory depth} of the channel $T$. Equivalent to this definition is to say that the \emph{memory branch}, i.e. the channel $T_{\mathcal{M}}:\mathcal{S}(\mathcal{M})\otimes\mathcal{S}(\mathcal{A}^{\otimes n})\mapsto \mathcal{S}(\mathcal{M})$ defined by $T_{\mathcal{M}}(\sigma_{\mathcal{M}}\otimes\sigma_{sys}):=\tr_{\mathcal{B}^{\otimes n}}[T_n(\sigma_{\mathcal{M}}\otimes\sigma_{sys})]$, completely depolarizes the information of the memory input state, see fig. \ref{fig:memory_2}.
\begin{figure}
\centering
	\includegraphics{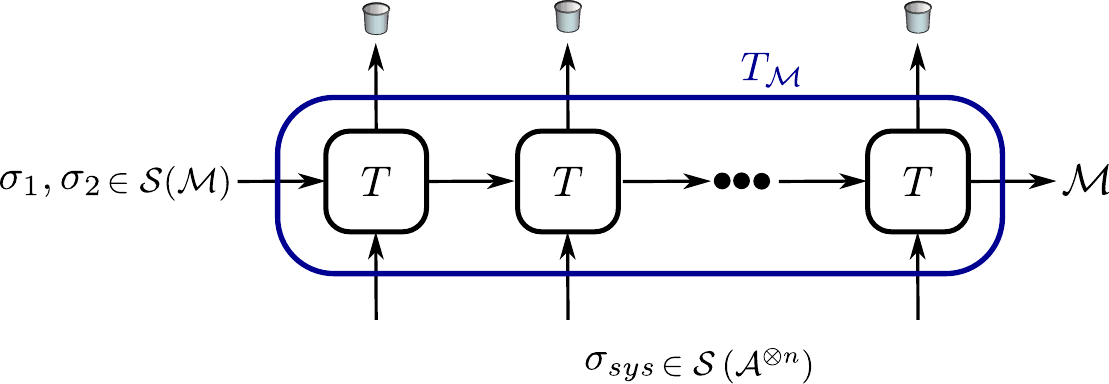}
	\caption{To characterize the strict forgetfulness of a memory channel $T$, we consider the channel's n-fold concatenation, where we neglect the output system $\mathcal{B}^{\otimes n}$ (depicted by the bins). The memory channel is strictly forgetful, iff there is an $n\in\mathbbm{N}$, such that the output on the memory system for any two different input states cannot be distinguished via an arbitrary measurement.}
	\label{fig:memory_2}
\end{figure}
Our aim is to construct strictly forgetful memory channels $T$ with memory depth $n$, exploiting the results about $n$-th order roots of the CDC. Similarly to the section about the construction of maximal roots in arbitrary dimensions, we construct the forgetful memory channels by expressing the problem in terms of matrix equations.

For this purpose we fix bases of operators on the memory channels input, output and memory system, i.e., $\{M_{i}\}$ is a basis for $\mathcal{M}$ and $\{A_i\}$ respectively $\{B_i\}$ are bases of $\mathcal{A}$ respectively $\mathcal{B}$, where the number $i\in\{1,...,d^2_X\}$ of operators correspond the respective dimension of the Hilbert space. The matrix representation of a memory channel $T$ is then given through
\begin{equation}
\bra{i,j}{D_{T}}\ket{k,l}:= \tr({M^i}^{*}\otimes {B^j}^{*} \ T(M_k\otimes A_l)),
\end{equation}
In what follows, we try to identify the parts of the matrix, which determine the forgetfulness of the corresponding memory channel $T$. For that purpose, we first consider the identity
\begin{align}
\tr_{\mathcal{B}^{\otimes 2}}\left[T_2(\sigma_{\mathcal{M}}\otimes\sigma_1\otimes\sigma_2)\right]&=\tr_{\mathcal{B}^{\otimes 2}}\left[\Big( id_{\mathcal{B}}\otimes T\Big)\circ\Big(T\otimes id_{\mathcal{B}}\Big)(\sigma_{\mathcal{M}}\otimes\sigma_1\otimes\sigma_2)\right]\\
&=\tr_{\mathcal{B}}\Big[T\Big(\tr_{\mathcal{B}}\left[T\left(\sigma_{\mathcal{M}}\otimes\sigma_1\right)\right]\otimes\sigma_2\Big)\Big]\, .\nonumber
\end{align}
If we apply this identity to the definition of the memory branch $T_{\mathcal{M}}(\sigma_{\mathcal{M}}\otimes\sigma_{sys})$ of the $n$-th concatenation acting on a separable input state , i.e. $\sigma_{sys}=\sigma_1\otimes...\otimes{\sigma_n}$, we find the following term:
\begin{align} \label{eq:concatenation_split} \notag
T_{\mathcal{M}}(\sigma_{\mathcal{M}}\otimes\sigma_{sys})&=\tr_{\mathcal{B}^{\otimes n}}\left[T_n(\sigma_{\mathcal{M}}\otimes{\sigma_1}\otimes...\otimes{\sigma_n})\right]\\
&=\tr_{\mathcal{B}}[T(\tr_{\mathcal{B}}[T(...T( \tr_{\mathcal{B}}\left[T(\rho_{\mathcal{M}}\otimes\sigma_1)\right]\otimes\sigma_2)\otimes...\otimes \sigma_n)].
\end{align}
By introducing the set of parameterized channels $T_{\mathcal{M},\sigma_i}:\mathcal{S}(\mathcal{M})\mapsto\mathcal{S}(\mathcal{M})$ on the memory branch, where $T_{\mathcal{M},\sigma_i}(\sigma_{\mathcal{M}}):=\tr_{\mathcal{B}}\left[T(\sigma_{\mathcal{M}}\otimes\sigma_i)\right]$ with $\sigma_i \in \mathcal{S}(\mathcal{A})$, we can rewrite \eqref{eq:concatenation_split} as concatenation of parameterized channels on the memory branch:
\begin{equation}\label{eq:parameterized}
T_{\mathcal{M}}(\sigma_{\mathcal{M}}\otimes\sigma_{sys})=T_{\mathcal{M},\sigma_{n}}\circ...\circ T_{\mathcal{M},\sigma_1}(\sigma_{\mathcal{M}}).
\end{equation}
We point out that so far we have just reformulated the description of the memory branch in terms of parametrized maps. If we now identify every $T_{\mathcal{M},\sigma_i}$ with its matrix representation $D_{T_{\mathcal{M},\sigma_{i}}}$ via the coefficients
\begin{equation}
\bra{k}D_{T_{\mathcal{M},\sigma_{i}}}\ket{l}:=\tr({M^k}^* \tr_{\mathcal{B}}[T(M_l\otimes\sigma_{i})]),
\end{equation}
for some basis of operators $\{M_k\}$ on the memory system, we can express \eqref{eq:parameterized} as multiplication of parametrized matrices:
\begin{equation}\label{eq:conc}
D_{T_{\mathcal{M}}}=D_{T_{\mathcal{M},\sigma_n}}\cdot D_{T_{\mathcal{M},\sigma_{n-1}}}\cdot\ldots \cdot D_{T_{\mathcal{M},\sigma_1}}
\end{equation}
Note that the left-hand-side of this equation implicitly depends on the system state $\sigma_{sys}$. Equation \eqref{eq:conc} turns out to be the crucial matrix equation to construct channels of memory depth $n$ utilizing the results of $n$-th order CDC-roots. Indeed, by the definition of a strictly forgetful channel $T$, the left-hand side needs to represent a completely depolarizing channel. Obviously, if the memory depth of $T$ is $n$, then $T_{\mathcal{M},\sigma}$ needs to be a root of a CDC of order $n_\sigma\leq n$ for all $\sigma\in\S(\mathcal{A})$ since we may choose $\sigma_i = \sigma$ for all $i$ in \eqref{eq:conc}. However, it is not enough to demand that $T_{\mathcal{M},\sigma}$ is a root of a CDC for all $\sigma\in\S(\mathcal{A})$ in order to construct a strictly forgetful memory channel. Indeed, there exist memory channels which are strictly forgetful for all system states of the form $\sigma_{sys}=\sigma^{\otimes n}$ but not for general system states, see the example at the end of this section.

The main obstacle to overcome is now that \eqref{eq:conc} needs to be completely depolarizing for all choices of the $\sigma_i$. Before we tackle this problem, let us argue that strict forgetfulness for all separable states of the input system $\mathcal{A}^{\otimes n}$ implies strict forgetfulness for all elements of  $\S(\mathcal{A}^{\otimes n})$. Suppose $T$ is strictly forgetful for all separable states $\sigma_{sys}=\sigma_1\otimes...\otimes{\sigma_n}$ and choose an operator basis $\{\sigma_\alpha\}_{\alpha=1,\ldots,d_{\mathcal{A}}^2}$ of $\mathcal{A}$ such that each $\sigma_\alpha$ is a quantum state. An arbitrary, possibly entangled, state $\rho_{sys}$ can then be written as
\begin{equation}
\rho_{sys}=\sum_{\alpha_1,\ldots,\alpha_n}c_{\alpha_1\ldots\alpha_n}\sigma_{\alpha_1}\otimes\ldots\otimes\sigma_{\alpha_n}\, .
\end{equation}
Let $\sigma_{\mathcal{M},1}$ and $\sigma_{\mathcal{M},2}$ be arbitrary states of the memory, then we get
\begin{align}
\tr_{\mathcal{B}^{\otimes n}}\left[T_n(\sigma_{\mathcal{M},1}\otimes\rho_{sys})\right]&=&\sum_{\alpha_1,\ldots,\alpha_n}c_{\alpha_1\ldots\alpha_n}\tr_{\mathcal{B}^{\otimes n}}\left[T_n(\sigma_{\mathcal{M},1}\otimes \sigma_{\alpha_1}\otimes\ldots\otimes\sigma_{\alpha_n}\right]\\
&=&\sum_{\alpha_1,\ldots,\alpha_n}c_{\alpha_1\ldots\alpha_n}\tr_{\mathcal{B}^{\otimes n}}\left[T_n(\sigma_{\mathcal{M},2}\otimes \sigma_{\alpha_1}\otimes\ldots\otimes\sigma_{\alpha_n}\right] \nonumber \\
&=&\tr_{\mathcal{B}^{\otimes n}}\left[T_n(\sigma_{\mathcal{M},2}\otimes\rho_{sys})\right]\, , \nonumber
\end{align}
and hence $T$ is also strictly forgetful for all $\rho\in\S(\mathcal{A}^{\otimes n})$. Thus, we can restrict to separable states $\sigma_{sys}$ without loss of generality, which means that everything boils down to assuring that \eqref{eq:conc} equals a CDC. Hence, we need to have a closer look at the parametrized matrices $T_{\mathcal{M},\sigma_{i}}$. To facilitate the derivation we choose the bases $\{M_i\},\{A_i\}$ and $\{B_i\}$ to be hermitian (so the dual bases) and the identity as the first element for each of them. We then find for some fixed $\sigma_i$:
\begin{align}
\bra{k}D_{T_{\mathcal{M},\sigma_{i}}}\ket{l}&=\tr_{\mathcal{M}}(M^k \ \ \tr_{\mathcal{B}}[T(M_l\otimes\sigma_{i})])\\
&=d_{\mathcal{B}}\cdot  \tr_{\mathcal{MB}}(M^k\otimes \underbrace{\frac{1}{d_\mathcal{B}}\mathbbm{1}_{\mathcal{B}}}_{=B^1} \cdot T(M_l\otimes\sum_{r=1}^{d_\mathcal{A}^2}\underbrace{\tr(A^r \sigma_i)}_{=:\alpha_r(\sigma_i)} A_r))\nonumber \\
&=\frac{d_{\mathcal{B}}}{d_{\mathcal{A}}}\cdot \underbrace{\bra{k,1}{D_{T}}\ket{l,1}}_{=:\bra{k}X_{1,1}\ket{l}}+d_{\mathcal{B}}\cdot \sum_{r=2}^{d_{\mathcal{A}}^2}\alpha_r(\sigma_i)\underbrace{\bra{k,1}{D_{T}}\ket{l,r}}_{=:\bra{k}X_{1,r}\ket{l}}\, .\nonumber
\end{align}
This equation shows that for appropriate bases every matrix on the right-hand side of \eqref{eq:conc} can be expressed as composition of the submatrix $X_{1,1}$ plus some state specific weighted sum of submatrices $\{X_{1,2...d^2}\}$ of the matrix $D_{T}$. As already mentioned, the matrix
\begin{equation}
\label{eq:memchan}
D_{T_\mathcal{M},\sigma}=\frac{d_{\mathcal{B}}}{d_{\mathcal{A}}}\cdot X_{1,1}+d_{\mathcal{B}}\cdot\sum_{r=2}^{d_\mathcal{A}^2}\alpha_r(\sigma) X_{1,r}
\end{equation}
must necessarily represent a root of a CDC of order at most $n$ for arbitrary $\sigma$. Moreover, all $n$-fold products of matrices $D_{T_\mathcal{M},\sigma_i}$ with arbitrary $\sigma_i$'s must also represent the CDC. This can be assured by fixing bases and choosing the matrix blocks according to
\begin{equation}\label{eq:submatrices}
X_{1,l}=\left(
    \begin{array}{cc}
     \frac{d_\mathcal{A}}{d_\mathcal{B}}\delta_{1,l} & 0 \\
      \frac{d_\mathcal{A}}{d_\mathcal{B}} v\delta_{1,l} &  J_l \\
    \end{array}
  \right) \in M_{d_{\mathcal{M}}^2}(\mathbb{C})\, ,
\end{equation}
where the $J_l$ are upper triangular matrices of dimension $d_\mathcal{M}^2-1$ and $v$ is a real valued vector. By choosing those blocks appropriately it can be assured that the nilpotency order of the $J_l$ is $n$, which, by our results on maximal roots of the CDC, is bounded from above by $n\leq d_\mathcal{M}^2-1$. This choice ensures that every $D_{T_{\mathcal{M},\sigma_{i}}}$ in \eqref{eq:conc} is of the form
\begin{equation}
D_{T_{\mathcal{M},\sigma_{i}}}=\left(
    \begin{array}{cc}
      1 & 0 \\
       v\ & \sum_{l=1}^{d_\mathcal{A}^2}\alpha_l(\sigma_i)J_l  \\
    \end{array}
  \right).
\end{equation}
 If we put this into the right-hand side of \eqref{eq:conc} we find that the memory branch is indeed completely depolarizing for an arbitrary input state of the form $\sigma_{sys}=\sigma_1\otimes...\otimes \sigma_n$ and the necessary number of concatenation steps is upper bounded by $d_\mathcal{M}^2-1$.

What remains is to show that complete positivity is not violated if we choose the sub-matrices $X_{1,k}$ in the proposed way. Here we emphasize that the behavior of the memory branch is just affected by the submatrices $X_{1,k}$ and we are completely free to choose $X_{2...d^2,k}$ to assure complete positivity. Moreover, we are free to choose the matrix blocks $J_l$ and $v$ with arbitrarily small but non-zero norm, without disturbing the forgetfulness property of $T$. Hence, the matrix blocks $J_l$ and $X_{2...d^2,k}$ can be considered as perturbation of the bistochastic completely depolarizing channel $T_{\idty}$ on $\mathcal{A}$, $\mathcal{B}$ and $\mathcal{M}$ defined via $T_{\idty}(\sigma_{\mathcal{MA}})=\frac{1}{d_\mathcal{B}d_\mathcal{M}}\idty_\mathcal{BM}$ for all $\sigma_{\mathcal{MA}}\in \S(\mathcal{MA})$.

A natural question to ask is whether this construction yields all strictly forgetful memory channels or if there are examples which cannot be transformed to the case of upper triangular matrices by clever choice of a basis. It turns out that our construction indeed covers all strictly forgetful memory channels.
\begin{theorem}
\label{thm:memchan}
Let $T:\mathcal{M}\otimes\mathcal{A} \mapsto  \mathcal{B}\otimes\mathcal{M}$ be a strictly forgetful memory channel. For appropriate bases of $\mathcal{A},\mathcal{B}$ and $\mathcal{M}$ the memory branch is of the form \eqref{eq:submatrices}. This implies that the memory depth of strictly forgetful memory channels is upper bounded by $d_\mathcal{M}^2-1$.
\end{theorem}
\begin{proof}
We choose again bases $\{A_i\},\{B_i\}$ and $\{M_i\}$ of $\mathcal{A},\mathcal{B}$ and $\mathcal{M}$ such that the identity is the first element of the respective basis and the other elements of the basis are hermitian and tracefree. We adopt the notation of \eqref{eq:memchan} and denote the matrices with elements $\bra{k,1}{D_{T}}\ket{l,r}$ by $X_{1,r}$. Let $\widehat X_{1,r}$ denote the matrix obtained from $X_{1,r}$ by deleting first row and column. The first step in our proof is to verify that the memory channel $T$ is strictly forgetful with memory depth at most $n$ iff the matrix algebra generated by the matrices $\widehat X_{1,l}$ is nilpotent. Since we have chosen the basis $\{A_i\}$ hermitian, tracefree and $A_1=\idty$, there are positive numbers $r_l$ such that
\begin{equation}
\sigma=\frac{1}{d}A_1 +\sum_{l=2}^{d^2_{\mathcal{A}}}a_l A_l
\end{equation}
is a quantum state for all $\vert a_l\vert\leq r_l$. Thus, the condition that $T_{\mathcal{M},\sigma_n}\circ \ldots \circ T_{\mathcal{M},\sigma_1}$ is completely depolarizing for all $\sigma_1,\ldots,\sigma_n$ implies that
\begin{equation}
\sum_{l_1,\ldots ,l_n}a_{l_1}\ldots a_{l_n}\widehat X_{1,l_1}\cdot\ldots\cdot \widehat X_{1,l_n} =0\, ,
\end{equation}
where the $a_{l_i}$ equal $1/d$ if $l_i=0$ and $\vert a_{l_i}\vert \leq r_{l_i}$ otherwise. If we consider this as a polynomial in the variables $a_{l_i}$ with matrix-valued coefficients we see that this equation implies that all coefficients must vanish, that is, $\widehat X_{1,l_1}\cdot\ldots\cdot \widehat X_{1,l_n}=0$ for all $l_i$. This proves that the algebra generated by the matrices $\widehat X_{1,l}$ is nilpotent.

By the theorem of Jacobson \cite{jacobson} this already implies that this algebra is simultaneously triangularizable, that is, there is a basis in which all matrices are upper triangular. The statement about the maximal memory depth of $T$ follows trivially.
\end{proof}
The crucial point in the proof of theorem \ref{thm:memchan} is to show that the algebra generated by the $\widehat X_{1,l}$ is nilpotent. If were only able to prove that the subspace generated by the $\widehat X_{1,l}$ is nilpotent, which translates into the property that the memory channel $T$ is strictly forgetful for all system states of the form $\sigma_{sys}=\sigma^{\otimes n}$, we could not conclude that $T$ is strictly forgetful. In fact, there exist examples \cite{mathes} of nilpotent subspaces of matrices which are not simultaneously upper triangular. From such an example it is easy to construct a memory channel which is strictly forgetful for all $\sigma_{sys}=\sigma^{\otimes n}$ but not for general system states. Indeed, let all systems $\mathcal{A},\mathcal{B}$ and $\mathcal{M}$ be qubits and choose Pauli matrices as operator basis. Consider the following matrices
\begin{equation}
\widehat X_{1,2}=
\left(
\begin{array}{ccc}
 0 & 0  & 0   \\
 -a & 0  & 0   \\
 0 & a  & 0
\end{array}
\right)
\quad
 \widehat X_{1,3}=
\left(
\begin{array}{ccc}
 0 & b  & 0   \\
 0 & 0  & b   \\
 0 & 0  & 0
\end{array}
\right)\quad
 \widehat X_{1,1}= \widehat X_{1,4}=0
\end{equation}
and assume all other matrix elements of $D_T$ to be zero, except $\bra{1}X_{1,1}\ket{1}=1$ which represents the trace-preserving property of $T$. Again, for small enough $a$ and $b$ this is completely positive. For a state $\sigma \in \S(\mathcal{A})$ we get the memory channel
\begin{equation}
D_{T_{\mathcal{M},\sigma}}=
\left(
\begin{array}{cccc}
 1 & 0  & 0  & 0 \\
 0 & 0  & b\alpha_y(\sigma)   & 0 \\
 0 &  -a\alpha_x(\sigma) & 0  & b\alpha_y(\sigma) \\
 0 &  0 & a\alpha_x(\sigma)   & 0
\end{array}
\right)\, ,
\end{equation}
where $\alpha_x(\sigma)$ respectively $\alpha_y(\sigma)$ denote the coefficients of $\sigma$ with respect to Pauli operator $x$ respectively $y$. Obviously, $D_{T_{\mathcal{M},\sigma}}$ is a root of the bistochastic CDC of order at most three, but the family of all $D_{T_{\mathcal{M},\sigma}}$ is not simultaneously upper triangular. This is expressed by the fact that the sequence
\begin{equation}
\left(D_{T_{\mathcal{M},\psi_x}} D_{T_{\mathcal{M},\psi_y}}\right) ^n= \left(
\begin{array}{cccc}
 1 & 0  & 0  & 0 \\
 0 & 0  & 0  & 0 \\
 0 &  0 & (-ab)^n  & 0 \\
 0 &  0 & 0  & (ab)^n
\end{array}
\right)\, ,
\end{equation}
where $\psi_x$ respectively $\psi_y$ denote the eigenstates of Pauli operators $x$ respectively $y$ with eigenvalue $+1$, never exactly equals the CDC, although it converges exponentially in $n$ towards the CDC.

\section*{Discussion}

Our construction implies that the maximal memory depth of the channel just depends on the dimension of the memory system. Under further assumptions the bound can sometimes be improved. For example, if the memory channel is assumed to be reversible, and hence given by a unitary operator, the memory branch is a homomorphism. Reversible qubit channels have been discussed in \cite{rybar}, and the maximal memory depth was shown to be $2<3=2^2-1$. More generally, one can see that the nesting of linear subspaces implicit in the Jordan decomposition has to be replaced in the reversible case by a nesting of subalgebras. Since for these some dimensions are forbidden, one gets a tighter bound, namely depth $<2(d-1)$ \cite{GRW}.

\section*{Acknowledgments}
We gratefully acknowledge financial support by the EU (projects CORNER and COQUIT) and stimulating conversations with David Gross and Thomas Salfeld.

\bibliographystyle{alpha}
\bibliography{channelBib}

\end{document}